\documentclass{jgaa-art}

\usepackage{booktabs}
\usepackage{amssymb}
\usepackage{amsmath}
\usepackage{graphicx}
\usepackage[subrefformat=parens,labelfont=default,font=small]{subcaption}
\usepackage{mathtools}
\mathtoolsset{showonlyrefs=true}

\usepackage{hyperref}
\usepackage[inline]{enumitem}
\usepackage{comment}
\usepackage{multirow}
\usepackage{dcolumn}

\newcommand{\mypiang}{\operatorname{ang_{\pi}}}
\newcommand{\maximize}{\operatorname{Maximize\quad}}

\DeclareMathOperator{\la}{la}
\DeclareMathOperator{\tw}{tw}
\DeclareMathOperator{\bt}{bt}
\DeclareMathOperator{\bw}{bw}
\DeclareMathOperator{\sep}{sep}

\DeclareMathOperator{\myseg}{seg}
\DeclareMathOperator{\myarc}{arc}

\graphicspath{{figures/}}
\newcommand{\nv}[1]{#1}

{\itshape}{\rmfamily}
{\itshape}{\rmfamily}
\newtheorem{corollary}{Corollary}{\bfseries}{\itshape}
{\bfseries}{\itshape}
\newtheorem{example}{Example}{\itshape}{\rmfamily}
{\itshape}{\rmfamily}
{\itshape}{\rmfamily}
{\itshape}{\rmfamily}
{\itshape}{\rmfamily}
\newtheorem{proposition}{Proposition}{\bfseries}{\itshape}
{\itshape}{\rmfamily}
{\itshape}{\rmfamily}
{\itshape}{\rmfamily}

\begin{document}
\doi{} 
\HeadingAuthor{Kryven et al.} 
\HeadingTitle{Drawing Graphs on Few Circles and Few Spheres}
\title{Drawing Graphs \\ on Few Circles and Few Spheres} 
\Ack{A preliminary version of this work appeared in Proc.\ 4th Conf.\
  Algorithms \& Discrete Appl.\ Math.\ (CALDAM'18), volume 10743 of
  Lect.\ Notes Comp. Sci., pages 164--178, Springer-Verlag.  M.K.\
  acknowledges support by DAAD.  A.R.\ acknowledges support by
  Erasmus+.  A.W.\ acknowledges support by DFG grant WO
  758/9-1.} 

\author[inst1]{Myroslav Kryven}{myroslav.kryven@uni-wuerzburg.de}
\author[inst2]{Alexander Ravsky}{alexander.ravsky@uni-wuerzburg.de}
\author[inst1]{Alexander Wolff}{orcid.org/0000-0001-5872-718X}

\affiliation[inst1]{Universit\"at W\"urzburg, Germany}
\affiliation[inst2]{Pidstryhach Institute for Applied Problems of
Mechanics and
Mathematics, National Academy of Sciences of Ukraine, Lviv, Ukraine}

\maketitle


\begin{abstract}
  Given a drawing of a graph, its \emph{visual complexity} is
  defined as the number of geometrical entities in the drawing, for
  example, the number of segments in a straight-line drawing or the
  number of arcs in a circular-arc drawing (in 2D).  Recently,
  Chaplick et al.\ [GD 2016] introduced a different measure
  for the visual complexity, the \emph{affine cover number}, which is
  the minimum number of lines (or planes) that together cover a
  crossing-free straight-line drawing of a graph $G$ in 2D (3D).
  In this paper, we introduce the \emph{spherical cover number}, which
  is the minimum number of circles (or spheres) that together cover a
  crossing-free circular-arc drawing in 2D (or 3D).  It turns out that
  spherical covers are sometimes significantly smaller than affine
  covers.
  For complete, complete bipartite, and platonic graphs,
  we analyze their spherical cover numbers and compare them to their
  affine cover numbers as well as their segment and arc numbers.  We
  also link the spherical cover number to other graph parameters such
  as
  treewidth and linear arboricity.
\end{abstract}

\section{Introduction}

A drawing of a given graph can be evaluated by many different quality
measures depending on the concrete purpose of the drawing.
Classical examples are the number of crossings, the ratio between the
lengths of the shortest and the longest edge, or the angular
resolution.
Clearly, different layouts (and layout algorithms) optimize different
measures.  Hoffmann et al.~\cite{hkkr-qrmgd-CCCG14} studied ratios
between optimal values of quality measures implied by different graph
drawing styles.  
\nv{For example, there is a circular-arc drawing of the icosahedron with
perfect angular resolution (that is, the edges are equiangularly
spaced around each vertex), whereas the best
straight-line drawing has an angular resolution of at
most~$15^{\circ}$, which yields a ratio of~$72^{\circ}/15^{\circ}=4.8$.
Hoffmann et al.\ also
constructed a family of graphs whose straight-line drawings have
unbounded edge--length ratio, whereas there are
circular-arc drawings with edge--length ratios arbitrarily close
to 3~\cite[Figures~4 and~6]{hkkr-qrmgd-CCCG14}. }

A few years ago, a new type of quality measure was introduced: the
number of geometric objects that are needed to draw a graph given a
certain style.  Schulz~\cite{s-dgfa-JGAA15} \nv{coined} this measure
the
\emph{visual complexity} of a drawing.  More concretely, Dujmovi{\'c}
et al.~\cite{Dujmovic2007DrawingsOP} defined the \emph{segment number}
$\myseg(G)$ of a graph~$G$ to be the minimum number of straight-line
segments over all straight-line drawings of~$G$.  Similarly,
Schulz~\cite{s-dgfa-JGAA15} defined the \emph{arc number} $\myarc(G)$
with respect to circular-arc drawings of~$G$ and showed that
circular-arc drawings are an improvement over straight-line drawings
not only in terms of visual complexity but also in terms of area
consumption; \nv{see Schulz~\cite[Theorem~1]{s-dgfa-JGAA15}}.
Mondal et al.~\cite{Mondal2013} showed how to minimize
the number of segments in convex drawings of 3-connected planar
graphs \nv{both on and off} the grid. 
Igamberdiev et al.~\cite{ims-dpc3c-JGAA17}
fixed a bug in the algorithm of Mondal et al.\ and compared the
resulting algorithm to two other algorithms in terms of angular
resolution, edge length, and face
aspect ratio.
\nv{H\"ultenschmidt et al.~\cite{hkms-dpgfg-JGAA18} studied the visual
complexity of drawings of planar graphs.  For example, they showed upper
bounds for the number of segments and arcs in drawings of trees,
triangulations, and general planar graphs.}
Recently, Kindermann et al.~\cite{kms-eaadfs-JGAA18}
presented a user study showing that people without mathematical or
computer science background prefer drawings that consist of few line
segments, that is, drawings of low visual complexity.
\nv{(Users with such a background had a slight tendency to prefer
drawings
that are more symmetric.)  The study, however, was done for trees
only.}

\nv{%
Durocher et al.~\cite{durocher} investigated the complexity
of computing minimum-segment drawings (and related problems).
Among others, they showed that it is NP-hard to compute the segment
number of plane graphs (that is, planar graphs with fixed
embedding), even if the graphs have maximum degree~4.
As an open problem, the authors suggested to study \emph{minimum-line
  drawings}, which they define to be minimum-segment drawings whose
edges lie in the union of the smallest number of straight lines
(among all minimum-segment drawings).}

\nv{Chaplick et al.~\cite{cflrvw-dgflf-GD16} defined a similar quality
  measure, which they call the \emph{affine cover number}.}  Given a
graph~$G$ \nv{and two integers~$l$ and~$d$ with $0<l<d$}, they
defined~$\rho_d^l(G)$ to be the minimum
number of $l$-dimensional affine subspaces that together cover a
crossing-free straight-line drawing of~$G$ in $d$-dimensional space.
 \nv{It turned out that it
  suffices to consider $l\le2$ because otherwise
  $\rho^l_d(G)=1$. In~\cite{cflrvw-dgflf-GD16} the authors also show
  that every graph can be drawn in~$3$-space as effectively as in high
  dimensional spaces, i.e., for any integers $1 \le l \le 3\le d$ and
  for any graph $G$, it holds that $\rho^l_d(G)=\rho^l_3(G)$.
  Note that, in general, the minimum-line drawings mentioned above are
  different from $\rho^1_2$-optimal drawings
  since there are graphs that do not have a $\rho^1_2$-optimal cover
  with the minimum number of segments; see
  Example~\ref{expl:rho_vs_ml} in Section~\ref{sec:discussion}.}

Among others, Chaplick et al.\ showed that the affine
cover number can be asymptotically smaller than the segment
number, constructing 
an infinite family of triangulations~$(T_n)_{n>1}$ such that $T_n$ has
$n$ vertices and $\rho^1_2(T_n)=O(\sqrt{n})$, but $\myseg(T_n)
=\Omega(n)$.  On the other hand, they showed that
$\myseg(G) = O(\rho^1_2(G)^2)$ for any connected planar graph $G$.
In a companion paper~\cite{cflrvw-cdgfl-WADS17}, Chaplick et al.\ show
that most variants of the affine cover number are NP-hard to compute.

\paragraph{Our contribution.}
Combining the approaches of Schulz and Chaplick et al.,
we introduce the \emph{spherical cover number} $\sigma_d^l(G)$ of a
graph~$G$ to be the minimum number of $l$-dimensional
spheres in $\mathbb{R}^d$ such that $G$ has a crossing-free
circular-arc drawing that is contained in the union of these
spheres.
\nv{Note that $\sigma_2^1(G)$ is defined for planar graphs only.}

\nv{Firstly, we provide some basic observations and preliminary
results that our work heavily relies on.}

We obtain bounds for the spherical cover number~$\sigma_3^2$ of the
complete and complete bipartite graphs which show that spherical
covers can be asymptotically smaller than affine covers; see
Table~\ref{tab:affine_cover_K} and Section~\ref{sec:complete}.

Then we turn to platonic graphs, that is, to 1-skeletons of platonic
solids; see Section~\ref{sec:platonic}.  These graphs possess several
nice properties: they are regular, planar and Hamiltonian.
We use them as indicators to compare the above-mentioned measures of
visual complexity; we provide bounds for their segment and arc numbers
(see Table~\ref{tab:seg_arc}) as well as for their affine and
spherical cover numbers (see Table~\ref{tab:affine_cover}).  For the
upper bounds, we
present straight-line drawings with (near-) optimal
affine cover number~$\rho_2^1$
and circular-arc
drawings with optimal spherical cover number~$\sigma_2^1$; see
Figures~\ref{fig:cube}--\ref{fig:icosahedron}.
We note that sometimes
optimal spherical covers are more symmetric than optimal affine covers. For example,
it seems that there is no symmetric drawing of the cube that is $\rho_2^1$-optimal,
whereas there are symmetric $\sigma_2^1$-optimal drawings; see
Fig.~\ref{fig:cube}.

For general graphs, we present lower bounds
for the spherical cover numbers by means of many combinatorial
graph characteristics, in particular, by the edge-chromatic number,
treewidth, balanced separator size, linear arboricity, and bisection
width; see Section~\ref{sec:lower}.

We decided to start with our more concrete (and partially stronger)
results and postpone the structural observations to
Section~\ref{sec:lower}, although this means that we'll sometimes have
to use forward references to Theorem~\ref{LowerSigma13}, our main
result in Section~\ref{sec:lower}.  Finally, we formulate an integer
linear program (MIP) that yields lower bounds for the segment number
of embedded planar graphs; see Section~\ref{sec:ilp}.  For the
platonic solids, the lower bounds (see Table~\ref{tab:piangles}) that
we computed using the MIP turned out to be tight.  We conclude with a
few open problems.

\section{Preliminary Results}
\label{sec:preliminaries}
In this section we state some preliminary results.
Firstly, we note that any drawing with straight-line segments and
circular arcs
can be transformed into a drawing that uses circular arcs only.
\begin{proposition}
  \label{prop:sigma-le-rho}
  Given a graph $G$ and a drawing~$\Gamma$ of~$G$ that represents
  edges as straight-line segments or circular arcs on $r$
  $l$-dimensional planes or spheres in~$\mathbb{R}^d$, there is a
  circular-arc drawing~$\Gamma'$ of $G$ on $r$ $l$-dimensional spheres
  in $\mathbb{R}^d$.  In particular, $\sigma_d^l(G)\le\rho_d^l(G)$ for
  any graph~$G$ and $1\le l<d$.
\end{proposition}
\begin{proof}
Take an arbitrary sphere $S\subset\mathbb{R}^d$
that does not intersect any of the $r$ spheres or planes that support
the given drawing~$\Gamma$ of~$G$.  Without loss of generality,
assume that $S$ is centered at the origin. \nv{This implies that none
of the spheres supporting~$\Gamma$ goes through the origin}.
Let~$\rho$ be the radius of~$S$.
Invert the drawing with respect to~$S$ by the map $x \mapsto \rho
x/\|x\|$.  The resulting drawing is a circular-arc drawing of~$G$ on
$r$ $l$-dimensional spheres in~$\mathbb{R}^d$.
Indeed, using basic properties of the inversion (see, for instance,~
\cite{e-ln-06} or~\cite[Chapter~5.1]{brannan_esplen_gray_2011}),
it can be proved that this inversion transforms planes into spheres of
the same dimension and preserves spheres, in other words, the set
of images of points on a sphere forms another sphere of the same
dimension.
\end{proof}
Therefore, we may consider any line a ``circle of infinite radius'',
any plane a ``sphere of infinite radius'', and any affine cover a
spherical cover.  By ``line'' we always mean a straight line.

Trivial bounds on $\sigma^1_3(G)$ follow from the fact that every
circle is contained in a plane and that we have more flexibility
\nv{when}
drawing in 3D than in 2D.
\nv{Note again that $\sigma^1_2(G)$ and $\sigma^1_3(G)$ are only defined
when $G$ is planar.}
\begin{proposition} \label{prop:sigma-ge-rho}
  For any graph~$G$, it holds that $\rho^2_3(G)\le\sigma^1_3(G)$.
  If~$G$ is planar, we additionally have
$\sigma^1_3(G)\le\sigma^1_2(G)$.
\end{proposition}

The spherical cover number $\sigma^2_3(G)$ can be considered a
characteristic of a graph $G$ that lies between its \emph{thickness}
$\theta(G)$, which is the smallest number of planar graphs whose union
is~$G$, and its \emph{book thickness} $\bt(G)$, also called page
number, which is the minimum number of pages (halfplanes) needed to
draw the edges of~$G$ when the vertices lie on the \emph{spine} of the
book (the line that bounds all halfplanes).

\begin{proposition}
  \label{prop:thickness}
  For every graph~$G$, it holds that $\theta(G) \le \sigma^2_3(G) \le
  \lceil \bt(G)/2\rceil$.
\end{proposition}

\begin{proof}
  Each sphere covers a planar subgraph of $G$, so $\sigma^2_3(G)$ is
  bounded from below by $\theta(G)$.
   On the other hand, given a book embedding of a graph $G$ with the
  minimum number of pages (equal to $\bt(G)$),
   we put the vertices from the spine along a circle which is the
  common intersection of $\lceil \bt(G)/2\rceil$ spheres;
  see Fig.~\ref{fig:sph_cov_K_n}.  Then, for
  each page, we draw all its edges as arcs onto a hemisphere.
  Thus,
  we obtain a drawing witnessing $\sigma^2_3(G)\le \lceil
  \bt(G)/2\rceil$.
\end{proof}

\nv{To bound $\sigma^1_2(G)$ and $\sigma^1_3(G)$ for the platonic
solids in Section~\ref{sec:platonic} from below we use a
combinatorial argument similar to that in Lemma~7(a) and Lemma~7(b)
in~\cite{cflrvw-dgflf-GD16}
which is based on the fact that each vertex
of degree at least 3 must be covered by at least two lines and
two lines can cross at most once, therefore, providing a lower
bound on the number of lines given the number of vertices. We use a
similar argument together with the fact that two circles can cross at
most twice.}
\begin{proposition}
\label{prop:combinatorial_argument}
For any integer $d \ge 1$ and any graph~$G$ with $n$ vertices
and $m$ edges, the
following bounds hold:
\begin{enumerate}[label=(\alph*)]
\item $\sigma^1_d(G) \geq \frac 12\left(1+\sqrt{1+2 \sum_{v\in
V(G)}\left\lceil\frac{\deg
v}{2}\right\rceil\left(\left\lceil\frac{\deg
v}2\right\rceil-1\right)}\right)$;
\label{thm:rho_mod_1}
\item $\sigma^1_d(G) \geq \frac 12\left(1+\sqrt{2m^2/n-2m+1}\right)$
for any graph~$G$ with $m\ge n\ge 1$;
\label{thm:rho_mod_2}
\end{enumerate}
\end{proposition}

\section{Complete and Complete Bipartite Graphs}
\label{sec:complete}

In this section we investigate the spherical cover numbers of complete
graphs and complete bipartite graphs.  We first cover these graphs by
spheres then by circular arcs, in 3D (and higher dimensions).

\begin{theorem}
  \sloppy ~~~~~~~~~~~~~

  \vspace{-1ex}
  \label{thm:sigma23K}
  \begin{enumerate}[label=(\alph*),itemsep=0pt]
    \item For any $n \ge 3$, it holds that $\lfloor
    (n+7)/6\rfloor\le\sigma^2_3(K_n)\le \lceil n/4\rceil$.
    \label{thm:sigma23K_Kn}

    \item For any $1 \le p \le q$, it holds that $pq/(2p+2q-4) \le
    \sigma^2_3(K_{p,q}) \le p$ and, if additionally $q>p(p-1)$,
    it holds that $\sigma^2_3(K_{p,q})=\lceil p/2\rceil$.
    \label{thm:sigma23K_Kpq}
  \end{enumerate}
\end{theorem}

\begin{proof}
  (a) By Proposition~\ref{prop:thickness}, $\theta(K_n)\le
  \sigma^2_3(K_n)\le \lceil \bt(K_n)/2\rceil$.  It remains to note
  that, e.g., Duncan~\cite{DUNCAN201195} showed that $\theta(K_n)\ge
  \lfloor (n+7)/6\rfloor$ and Bernhart and
  Kainen~\cite{BERNHART1979320} showed that $\bt(K_n)=\lceil
  n/2\rceil$.

  (b) Again, it suffices to bound the values of the graph's thickness
  and book thickness.  It can be easily shown that $\bt(K_{p,q})\le
  \min\{p,q\}$.  On the other hand, Harary~\cite[Section~7,
  Theorem~8]{harary} showed that $\theta(K_{p,q}) \ge
  pq/(2p+2q-4)$.  Due to Proposition~\ref{prop:thickness},
  $\theta(K_{p,q}) \le \sigma^2_3(K_{p,q}) \le \min\{p,q\} \le p$.  In
  particular, if $q>p(p-1)$ then $\bt(K_{p,q})=p$, due to Bernhart and
  Kainen~\cite[Theorem.~3.5]{BERNHART1979320}, and $\lceil
  pq/(2p+2q-4)\rceil=\lceil p/2\rceil$, so in this case
  $\sigma^2_3(K_{p,q})=\lceil p/2\rceil$.
\end{proof}

Theorem~\ref{thm:sigma23K} implies that any $n$-vertex graph~$G$
has $\sigma^2_3(G)\le \lceil n/4\rceil$.

\nv{On the other hand, given a graph $G$, we can bound $\sigma^1_3(G)$
  from below in terms of the \emph{bisection width} $\bw(G)$ of~$G$,
  that is, the minimum number of edges between the two sets
  $(W_1,W_2)$ of a \emph{bisection} of~$G$ that is, a partition of the
  vertex set~$V(G)$ of~$G$ into two sets~$W_1$ and~$W_2$ with $|W_1| =
  \lceil n/2 \rceil$ and $|W_2| = \lfloor n/2 \rfloor$.
\begin{proposition}
  \label{thm:rho_bw}
  For any graph $G$ and $d \ge 2$, it holds that $\sigma^1_d(G)\ge
  \bw(G)/2$.
\end{proposition}
\begin{proof}
  The proof is similar to the proof in Theorem~9(a)
  in~\cite{cflrvw-dgflf-GD16}.  It is based on the fact that for any
  finite set of points in $\mathbb{R}^d$, there is a hyperplane that
  bisects the point set into two almost equal subsets (that is, one
  subset may have at most one point more than the other).  Given a
  drawing of~$G$ with $\sigma$ arcs, a hyperplane bisecting~$V(G)$ can
  cross at most $2\sigma$ edges since a hyperplane can cross an arc at
  most twice.
\end{proof}
}

Next we analyze the bisection width of the complete (bipartite) graphs.

\begin{proposition}
  \label{lem:bwComplete}
  For any $n$, $p$, and $q$, $\bw(K_n)=\lfloor n^2/4 \rfloor$ and
  $\bw(K_{p,q})=\lceil pq/2\rceil$.
\end{proposition}

\begin{proof}
  Let $(W,W')$ be a bisection of $K_n$ such that $|W|=\lfloor
  n/2\rfloor$.   Then the width
  of this bisection is $\lfloor n^2/4 \rfloor$.  

  \nv{Now let $P \cup Q=V(K_{p,q})$ be the bipartition of~$K_{p,q}$,
    and let $(W,W')$ be a bisection of~$K_{p,q}$ that contains~$r$
    vertices from~$P$ and~$s$ vertices from~$Q$}
    (with $r+s=\lfloor(p+q)/2\rfloor$).  Then
  the width of this bisection is $r(q-s)+s(p-r)$. The minimum of this
  value can be found by a routine calculation of the minimum of a
   quadratic polynomial on the grid over the
  possible values of~$r$ and~$s$.
\end{proof}

\begin{theorem}
  \label{thm:sigma13K}
  For any positive integers $n$, $p$, and $q$, it holds that
  \begin{enumerate}[label=(\alph*),itemsep=0pt,topsep=1ex]
  \item $\lfloor n^2/8\rfloor \le \sigma^1_3(K_n) \le (n^2 + 5n +
    6)/6$ \label{thm:sigma13K_Kn} and
  \item $\lceil pq/4\rceil \le \sigma^1_3(K_{p,q}) \le \lceil
    p/2\rceil \lceil q/2\rceil$.\label{thm:sigma13K_Kpq}
  \end{enumerate}
\end{theorem}

\begin{proof}
  The lower bounds follow from Proposition~\ref{thm:rho_bw}
  and~\ref{lem:bwComplete}.

  To show the upper bound for $\sigma^1_3(K_n)$, we use a partition
  of~$K_n$ into (mutually edge-disjoint) subgraphs of~$K_3$ (that is,
  copies of $K_3$, paths of length~2, and single edges).  Using
  Steiner triple systems, one can show that $(n^2 + 5n + 6)/6$
  subgraphs suffice~\cite[Theorem~12]{cflrvw-dgflf-GD16}.
  \nv{For distinct
    points~$a$, $b$, and~$c$, let $L(a,b)$ be the line through~$a$
    and~$b$ and let $C(a,b,c)$ be the (unique) circle through~$a$,
    $b$, and~$c$.  For $n \le 3$, it is clear how to draw~$K_n$.}  
  For $n \ge 4$, we iteratively construct a set~$P$
  of $n$ points in~$\mathbb{R}^3$ satisfying the following conditions:
  \begin{itemize}
  \item no four distinct points of $P$ are coplanar,
  \item for any five distinct points $p_1,\dots,p_5\in P$, it holds
    that $C(p_1,p_2,p_5) \cap C(p_3,p_4,p_5) = \{p_5\}$ and
    $L(p_1,p_2) \cap C(p_3,p_4,p_5) = \emptyset$. 
  \item for any six distinct points $p_1,\dots,p_6\in P$, it holds
    that $C(p_1,p_2,p_3) \cap C(p_4,p_5,p_6) = \emptyset$.
  \end{itemize}
  It can be checked that these conditions forbid only a so-called
  \emph{nowhere dense set}
   of~$\mathbb{R}^3$ to place the next point of~$P$, so we can always
  continue.  Finally, we map the vertices of~$K_n$ to the distinct
  points of the set~$P$.  \nv{Consider our partition of $K_n$ into
  subgraphs of~$K_3$.  Each subgraph of~$K_3$ with at least two edges
  uniquely determines a circle or a circular arc, which we draw.  For
  each subgraph that consists of a single edge, we draw the line
  segment that connects the two vertices.}  The above conditions ensure
  that the drawings of no two subgraphs have a crossing.

  The upper bound for $\sigma^1_3(K_{p,q})$ can be seen as follows.
  Let $p'=\lceil p/2\rceil\ge p/2$ and $q'=\lceil q/2\rceil\ge
  q/2$. Draw a bipartite graph $K_{2p',2q'}\supset K_{p,q}$
  in~3D as follows; see Fig.~\ref{fig:sph_cov_K_pq}.  Let
  $V(K_{2p',2q'})=P\cup Q$ be the natural bipartition of its vertices.
  Fix any family of $p'$ distinct spheres with a common intersection
  circle. Place the $2q'$ vertices of~$Q$ on $q'$ distinct pairs of
  antipodal points on the circle. Consider a line going
  through the center of the circle and orthogonal to its plane. Place
  the $2p'$ vertices of~$P$ into $p'$ pairs of distinct intersection
  points of the line with the circles of the family, the points from
  each pair belonging to the same sphere.  Now each pair of antipodal
  points in~$Q$ together with each pair of cospheric points in~$P$
  determine a unique circle that contains all these points and
  provides a drawing of the four edges between them.  The union of all
  these circles is the desired drawing of~$K_{2p',2q'}$ onto $p'q'$
  circles.
\end{proof}

We remark that Proposition~\ref{prop:thickness} and all the bounds for
3D in this section and also hold for higher dimensions.

\begin{figure}
  \hfill
  \begin{subfigure}[b]{.65\textwidth}
    \centering
    \includegraphics{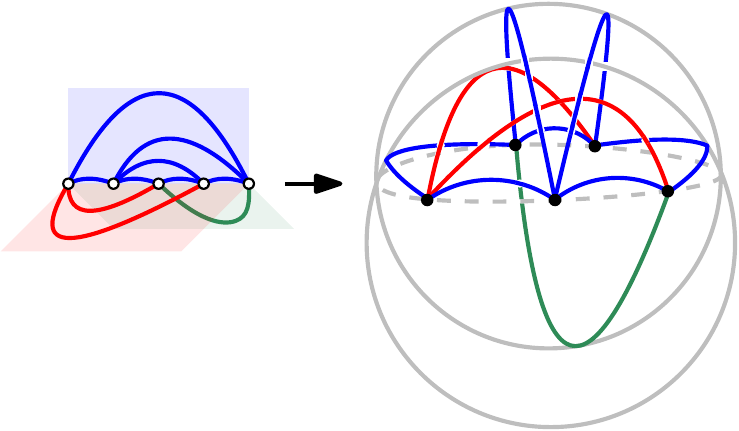}
    \caption{$\sigma_3^2(K_5) \le \lceil \theta(K_5)/2 \rceil = 2$}
    \label{fig:sph_cov_K_n}
  \end{subfigure}
  \hfill
  \begin{subfigure}[b]{.3\textwidth}
    \centering
    \includegraphics{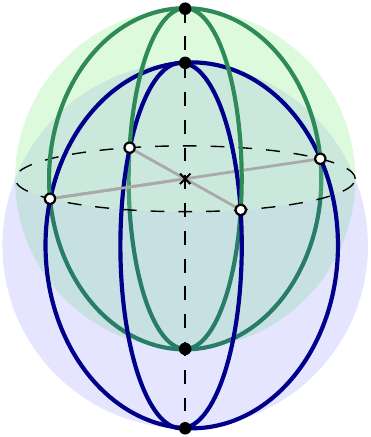}
    \caption{$\sigma_3^2(K_{4,4}) \le 2$}
    \label{fig:sph_cov_K_pq}
  \end{subfigure}
  \hfill\hspace*{0ex}

  \caption{Upper bounds for the spherical cover number of complete
    (bipartite) graphs.}
  \label{fig:spherical-cover}
\end{figure}

\begin{table}[htb]
  \centering
  $
  \begin{array}{lr@{\,\sim\,}lr@{\,\sim\,}lc}
    \toprule
    \phantom{\rho^1_2(}G  & \multicolumn{2}{c}{ K_n}  &  \multicolumn{2}{c}{ K_{p,q}}
 & \text{references}\\
\midrule

    \rho_3^1(G)   &  \multicolumn{2}{c}{ \binom{n}{2}}  &
    \multicolumn{2}{c}{ pq - \left\lfloor \frac{p}{2} \right\rfloor
    - \left\lfloor \frac{q}{2} \right\rfloor}
    & \text{\cite[Expl.~10 \&~25(c)]{cflrvw-dgflf-GD16}} \\[2ex]

    \rho_3^2(G) &  \frac{n^2-n}{12} &
    \frac{n^2 + 5n + 6}{6}
     &\multicolumn{2}{c}{  \left\lceil \frac{\min\{p,q\}}{2} \right\rceil} &
     \text{\cite[Thm.~12, Expl.~11]{cflrvw-dgflf-GD16}} \\

    \sigma_3^1(G)   & \left\lfloor \frac{n^2}{8} \right\rfloor &
    \frac{n^2 + 5n + 6}{6} & \left\lceil\frac{pq}4\right\rceil &
    \left\lceil\frac{p}{2}\right\rceil
    \left\lceil\frac{q}{2}\right\rceil
    & \text{Theorem~\ref{thm:sigma13K}}\\[2ex]

     \sigma_3^2(G) & \left\lfloor \frac{n+7}{6} \right\rfloor
    &  \left\lceil \frac{n}{4} \right\rceil &
    \left\lceil \frac{pq}{2(p+q-2)} \right\rceil &
    \left\lceil \frac{\min\{p,q\}}{2} \right\rceil
    &  \text{Theorem~\ref{thm:sigma23K}}\\

    \bottomrule
  \end{array}$

  \caption{Lower and upper bounds on the three-dimensional line,
    plane, circle, and sphere cover numbers of $K_n$ for any $n \ge 1$
    and of $K_{p,q}$ for any $p, q\ge 3$. \nv{The cells with only one
    entry contain tight bounds.}}
  \label{tab:affine_cover_K}
\end{table}

Table~\ref{tab:affine_cover_K} summarizes the known bounds for the
affine cover numbers~\cite{cflrvw-dgflf-GD16} and the new bounds for
the spherical cover numbers of complete (bipartite) graphs in~3D.

\section{Platonic Graphs}
\label{sec:platonic}

In this section we analyze the segment numbers, arc numbers, affine
cover numbers, and spherical cover numbers of platonic graphs.
\nv{We
provide upper bounds via the corresponding drawings; see
Figures~\ref{fig:tetrahedron}--\ref{fig:icosahedron}.}

\nv{To bound the spherical cover numbers~$\sigma_2^1$ and~$\sigma_3^1$
of the platonic graphs from below,} we use a single combinatorial
argument---Proposition~\ref{prop:combinatorial_argument}\ref{thm:rho_mod_1};
see Section~\ref{sec:lower}.  For the
affine cover number~$\rho_2^1$, a similar combinatorial argument
fails~\cite[Lemma~7(a)]{cflrvw-dgflf-GD16}.  Therefore, we
\nv{bound~$\rho_3^1$ (and, hence, also~$\rho_2^1$) from below} for
each platonic graph individually; see Proposition~\ref{prop:lines}.
For an overview of
our results, see Tables~\ref{tab:seg_arc} and~\ref{tab:affine_cover}.
We abbreviate every platonic graph by its capitalized initial; for
example, $C$ for the cube.

\begin{table}
\centering

\begin{tabular}{l@{\hspace{2ex}}rrr@{\hspace{3ex}}r@{\hspace{2ex}}c@{
\hspace{2ex}}r@{\hspace{2ex}}c@{\hspace{2ex}}c}
\toprule
$G=(V,E)$    & $|V|$ & $|E|$ & $|F|$ & $\myseg$&
  upp.$\,$bd. & $\myarc$ & lower bd. & upp.$\,$bd. \\
\midrule
tetrahedron  &  4 &  6 &  4 &  6 & Fig.~\ref{fig:tet_seg} &  3 &
Prop.~\ref{prop:combinatorial_argument}\ref{thm:rho_mod_1} &
\nv{Prop.~\ref{prop:arcs}\ref{prop:arcs_T}} \\
octahedron   &  6 & 12 &  8 &  9 & Fig.~\ref{fig:oct_seg} &  3 &
Prop.~\ref{prop:combinatorial_argument}\ref{thm:rho_mod_1} &
\nv{Prop.~\ref{prop:arcs}\ref{prop:arcs_O}} \\
cube         &  8 & 12 &  6 &  7 & Fig.~\ref{fig:cube_seg}&  4
& \cite[Lem.$\,$5]{Dujmovic2007DrawingsOP}      &
\nv{Prop.~\ref{prop:arcs}\ref{prop:arcs_C}} \\
dodecahedron & 20 & 30 & 12 & 13 & Fig.~\ref{fig:dod_seg} & 10
& \cite[Lem.$\,$5]{Dujmovic2007DrawingsOP}  &
\nv{Prop.~\ref{prop:arcs}\ref{prop:arcs_D}} \\
icosahedron  & 12 & 30 & 20 & 15 & Fig.~\ref{fig:ico_seg_opt}
& $7$& Prop.~\ref{prop:combinatorial_argument}\ref{thm:rho_mod_1} &
Prop.~\ref{prop:arcs}\ref{prop:ico_arcs} \\
\bottomrule
\end{tabular}

\caption{Bounds on the segment and arc numbers of the platonic graphs.
  We obtained the lower
  bounds on the segment number with the help of an integer linear
  program; see Table~\ref{tab:piangles} in Section~\ref{sec:ilp}.
  The upper bounds for the segment numbers of the dodecahedron and
  icosahedron have been established by Schulz~\cite{s-dgfa-JGAA15} and
  Scherm~\cite[Fig.~2.1(c)]{s-mukkgg-BTh16}.}
\label{tab:seg_arc}

\bigskip

\begin{tabular}{l@{\hspace{2ex}}r@{\hspace{2ex}}rl@{\hspace{1.5ex}}cr@
{\hspace{2ex}}r@{\hspace{2ex}}c}
\toprule
graph & $\rho_2^1$& $\rho_3^1$ & \multicolumn{1}{c}{lower~bd.~~} &
upp.$\,$bd.\ & $\sigma_2^1$ & $\sigma_3^1$ & upp.$\,$bd.\ \\
\midrule
tetrahedron & 6 & 6 & \cite[Expl.$\,$10]{cflrvw-dgflf-GD16}
& Fig.~\ref{fig:tet_seg} &  3 & 3 & Fig.~\ref{fig:tet_arc} \\

octahedron  & 9 & 9 & Prop.~\ref{prop:lines}\ref{prop:lines_O}
& Fig.~\ref{fig:oct_seg}  & 3 & 3 & Fig.~\ref{fig:oct_arc} \\

cube        & 7 & 7 & Prop.~\ref{prop:lines}\ref{prop:lines_C}
& Fig.~\ref{fig:cube_seg} & 4 & 4 & Fig.~\ref{fig:cube_arc} \\

dodecahedron & $9\dots10$ & $9\dots10$
& Prop.~\ref{prop:lines}\ref{prop:lines_D}
& Fig.~\ref{fig:dod_seg} & 5 & 5 & Fig.~\ref{fig:dod_circle_arc} \\

icosahedron  & $13\dots15$ & $13\dots15$
& Prop.~\ref{prop:lines}\ref{prop:lines_I}
& Fig.~\ref{fig:ico_seg_opt} & 7 & 7 & Fig.~\ref{fig:ico_circle_cover}
\\
\bottomrule
\end{tabular}

\caption{Bounds on the affine cover numbers $\rho_d^l$ and the
  spherical cover numbers $\sigma_d^l$ for platonic graphs.  The lower
  bounds on~$\sigma^1_2$ and~$\sigma^1_3$ stem from
  Proposition~\ref{prop:combinatorial_argument}\ref{thm:rho_mod_1}.}
\label{tab:affine_cover}
\end{table}

\begin{figure}[p]
  \begin{minipage}[b]{.265\textwidth}
    \begin{subfigure}[b]{.48\textwidth}
      \centering
      \includegraphics[page=2]{tetrahedron}
      \caption{6 segm.}
      \label{fig:tet_seg}
    \end{subfigure}
    \hfill
    \begin{subfigure}[b]{.48\textwidth}
      \centering
      \includegraphics[page=1]{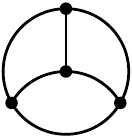}
      \caption{3 arcs}
      \label{fig:tet_arc}
    \end{subfigure}
    \caption{Drawings of the tetrahedron}
    \label{fig:tetrahedron}
  \end{minipage}
  \hfill
  \begin{minipage}[b]{.7\textwidth}
    \begin{subfigure}[b]{0.33\textwidth}
      \centering
      \includegraphics[scale=.95]{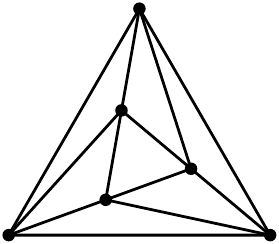}
      \caption{9$\,$segm.$\,$/$\,$lines}
      \label{fig:oct_seg}
    \end{subfigure}
    \hspace*{-1.5ex}
    \begin{subfigure}[b]{0.33\textwidth}
      \centering
      \includegraphics[scale=.95]{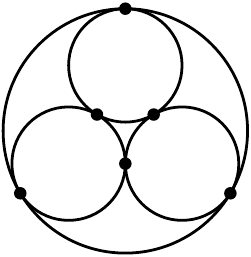}
      \caption{4 arcs$\,$/$\,$circles}
      \label{fig:oct_arc_orig}
    \end{subfigure}
    \begin{subfigure}[b]{0.3\textwidth}
      \centering
      \includegraphics[scale=.95]{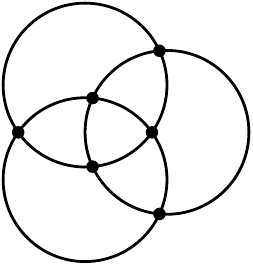}
      \caption{3 arcs$\,$/$\,$circles}
      \label{fig:oct_arc}
    \end{subfigure}
    \caption{Drawings of the octahedron}
    \label{fig:octahedron}
  \end{minipage}

\bigskip
\bigskip

  \begin{minipage}{\textwidth}
    \begin{subfigure}[b]{0.27\textwidth}
      \centering
      \includegraphics[scale=.95]{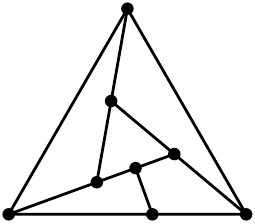}
      \caption{7 segm.$\,$/$\,$lines}
      \label{fig:cube_seg}
    \end{subfigure}
    \hfill
    \begin{subfigure}[b]{0.22\textwidth}
      \centering
      \includegraphics{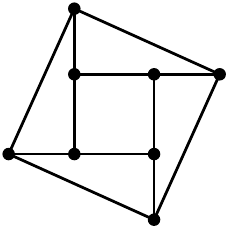}
      \caption{8 segm.$\,$/$\,$lines}
      \label{fig:cube_seg_orig}
    \end{subfigure}
    \hfill
    \begin{subfigure}[b]{0.26\textwidth}
      \centering
      \includegraphics[scale=.95]{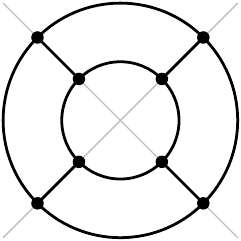}
      \caption{6$\,$arcs$\,$/$\,$4$\,$circles}
      \label{fig:cube_arc_orig}
    \end{subfigure}
    \hfill
    \begin{subfigure}[b]{0.22\textwidth}
      \centering
      \includegraphics[scale=.95]{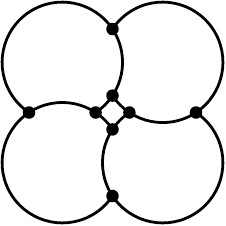}
      \caption{4 arcs}
      \label{fig:cube_arc}
    \end{subfigure}
    \caption{Drawings of the cube}
    \label{fig:cube}
  \end{minipage}

\bigskip
\medskip

  \begin{minipage}{\textwidth}
    \begin{subfigure}[b]{0.27\textwidth}
      \includegraphics[scale=.95]{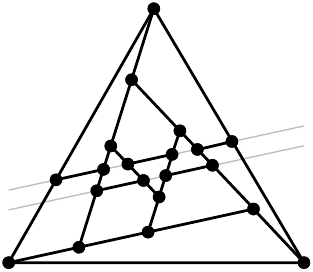}

\caption{13$\,$segm.$\,$/$\,$10$\,$lines$\,$\cite{s-mukkgg-BTh16}}
      \label{fig:dod_seg}
    \end{subfigure}~%
    \begin{subfigure}[b]{0.25\textwidth}
      \includegraphics[scale=.95]{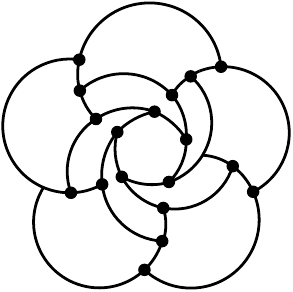}
      \caption{10$\,$arcs$\,$/$\,$10$\,$circ.$\,$\cite{s-dgfa-JGAA15}}
      \label{fig:dod_schulz}
    \end{subfigure}~%
    \begin{subfigure}[b]{0.25\textwidth}
      \includegraphics[scale=.95]{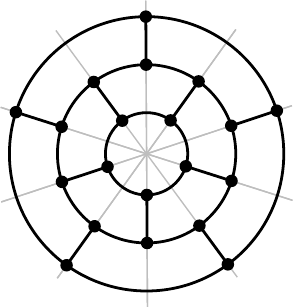}
      \caption{13$\,$arcs$\,$/$\,$8$\,$circ.$\,$\cite{s-mukkgg-BTh16}}
      \label{fig:dod_scherm}
    \end{subfigure}~%
    \begin{subfigure}[b]{0.22\textwidth}
      \includegraphics[scale=.95]{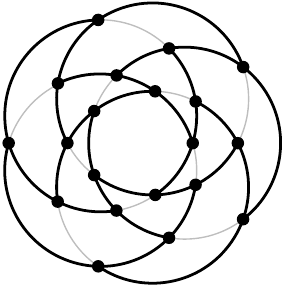}
      \caption{\mbox{10$\,$arcs$\,$/$\,$5$\,$circles}}
      \label{fig:dod_circle_arc}
    \end{subfigure}
    \caption{Drawings of the dodecahedron}
    \label{fig:dodecahedron}
  \end{minipage}

\smallskip

  \begin{minipage}{\textwidth}
    \begin{subfigure}[b]{0.3\textwidth}
      \centering
      \includegraphics{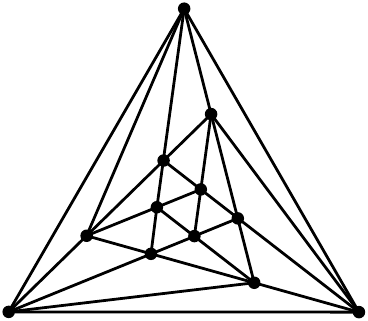}
      \caption{15 segments / lines}
      \label{fig:ico_seg_opt}
    \end{subfigure}
    \hfill
     \begin{subfigure}[b]{0.3\textwidth}
      \centering
      \includegraphics{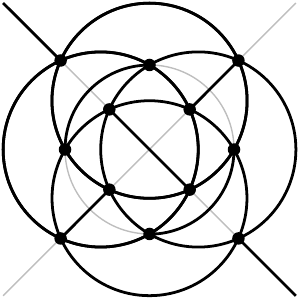}
      \caption{10 arcs / 7 circles}
      \label{fig:ico_arc}
    \end{subfigure}
    \hfill
    \begin{subfigure}[b]{0.3\textwidth}
      \centering
      \includegraphics{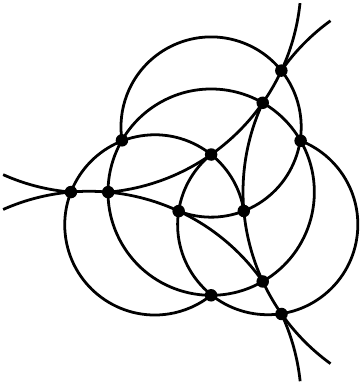}
      \caption{7 arcs / 7 circles}
      \label{fig:ico_circle_cover}
    \end{subfigure}
    \caption{Drawings of the icosahedron}
    \label{fig:icosahedron}
  \end{minipage}
\end{figure}

\begin{proposition}
\label{prop:lines}
\begin{enumerate*}[label=(\alph*)]
\item $\rho^1_3(T) \ge 6$;\label{prop:lines_T}
\item $\rho^1_3(O) \ge 9$;\label{prop:lines_O}
\item $\rho^1_3(C) \ge 7$;\label{prop:lines_C}
\item $\rho^1_3(D) \ge 9$;\label{prop:lines_D}
\item $\rho^1_3(I) \ge 13$.\label{prop:lines_I}
\end{enumerate*}
\end{proposition}

\noindent\emph{Proof.}~
  \ref{prop:lines_T} Follows from~\cite[Ex.~10]{cflrvw-dgflf-GD16}.

  \ref{prop:lines_O} Consider a straight-line drawing of the
  octahedron~$O$ covered
  by a family~$\mathcal L$ of $\rho$ lines.
  Observe that every vertex of the octahedron is adjacent to every
  other except the opposite vertex. Therefore, no line in
  $\mathcal{L}$ can cover more than three vertices,
  otherwise the edges on the line would overlap. Hence,
  every line covers at most two edges, and these must be adjacent.
  Moreover, the two end vertices of these length-2 paths cannot be
  adjacent. Since there are only three pairs of such vertices,
  at most three lines cover two edges each.
  Since the octahedron has twelve edges, $\rho \ge 9$.

  \ref{prop:lines_C} Now consider a straight-line drawing of the
  cube~$C$ covered
  by a family~$\mathcal L$ of $\rho$ lines.  We distinguish two cases.

  Assume first that the drawing of the cube lies in a single plane.
  Each embedding of the cube contains two nested
  cycles, namely, the boundary of the outer face and the innermost
  face.  We consider three cases depending on the shape of the outer
  face.
  (i)~\nv{If the outer cycle is drawn as the boundary of a (strictly)
  convex quadrilateral}, then
  none of the lines covering its sides can be used to cover
  the edges of the innermost cycle, therefore, it needs three
  additional lines.  (ii)~\nv{If the outer cycle is drawn
  as the boundary of a non-convex quadrilateral}, then we need
  three additional lines to cover the three edges going from
  its three convex angles to the innermost cycle.
  (iii)~Now assume that the outer cycle is drawn as a triangle. Then
  none of the lines covering its sides can be used to cover
  the edges of the innermost cycle.  If this cycle is drawn as a
  quadrilateral, then
  we need four additional lines to cover its sides.
  If the innermost cycle is drawn as a triangle, then we need three
  lines for the triangle and an additional line
  to cover the edge incident to the vertex of the innermost cycle
  which is not a vertex of the triangle.
  In each of the three cases (i)--(iii), we need at least seven lines
  to cover the cube.

  Now assume that the drawing of the cube is not contained in
  a single plane. Then its convex hull has (at least) four extreme
points.
  In order to cover the cube, we need at least one pair of
  intersecting   lines of~$\mathcal L$ for each vertex of the cube
  and at least three such pairs for each extreme point,
  that is, at least $4+4\cdot 3=16$  pairs of intersecting straight
  lines in total.  So, $\binom{\rho}{2}\ge 16$ and $\rho\ge 7$.

  \ref{prop:lines_D} Consider a straight-line drawing of the
  dodecahedron $D$ covered
  by a family~$\mathcal L$ of $\rho$ lines.  Again, we distinguish two
  cases.

  Assume first that the drawing of the dodecahedron lies in a
  single plane.  Again we make a case distinction
  depending on the shape of the outer cycle.
  \nv{(i)~If the outer cycle is drawn as the boundary of a convex polygon,
  let~$\mathcal L_0 \subseteq L$ be the family of lines that support
  the edges on the outer cycle.  This family consists of at least
  three lines.}  None of them covers any of the at most~15 vertices
  remaining in the interior of the convex polygon.  Thus each of these
  vertices is an intersection point of two lines of
  $\mathcal{L}\setminus\mathcal{L}_0$.  Since
  $\mathcal{L}\setminus\mathcal{L}_0\le\rho-3$, this family of
  lines can generate at most $\binom{\rho-3}{2}$ intersection
  points.  Therefore, $\binom{\rho-3}{2}\ge 15$ and, hence, $\rho\ge
9$.
  (ii)~Assume that the outer cycle is drawn as a non-convex
  quadrilateral.
  Then the drawing is contained in a convex angle opposite to the
  reflex angle. To cover the angle sides, we need a family~$\mathcal
  L_0$ consisting of
  at least two lines. None of them covers any of the
  at least $15+1$ vertices remaining in the interior of the angle.
  Similarly to the previous paragraph, we obtain
  $\binom{\rho-2}{2}\ge 16$ and, hence, $\rho\ge 9$.
  (iii)~\nv{Assume that the outer cycle is drawn as a
  pentagon $P$.}  Since
  the angle sum of a pentagon is~$3\pi$, $P$ has at most two reflex
  angles, and therefore, at least three convex angles. Each vertex of
  $D$ drawn as a vertex of a convex angle is an intersection point of
  (at least)
  three covering lines, because it has degree~3.
  There exists an edge~$e$ of~$P$ such that $P$ is contained
  in one of the half-planes created by the line~$\ell$
  spanned by~$e$ (see, for
  instance,~\cite{mmo-60}). It is easy to check that $\ell$ can cover
  only edge $e$ of
  the outer face of $D$. Then the family $\mathcal L\setminus\{\ell\}$
  covers all edges of $G$
  but~$e$. The angles of $P$ incident to $e$ are convex.
  Let $v$ be a vertex of $D$ drawn as a vertex of a convex angle not
incident to $e$.
  In order to cover~$D$, we need at least one pair of intersecting
  lines from $\mathcal L\setminus\{\ell\}$ for each vertex of~$D$
  different from~$v$
  and at least three such pairs for~$v$, that is, at least $19+3=22$
  pairs of intersecting lines in total.  Therefore,
  $\binom{\rho-1}{2}\ge 22$ and, hence, $\rho\ge 9$.
  Note that, in each of the three cases (i)--(iii), we have $\rho\ge
9$.

  Now assume that the drawing of~$D$ is not contained in a single
  plane.  Then its convex hull has (at least) four extreme points.
  In order to cover~$D$, we need at least one pair of intersecting
  lines of~$\mathcal L$ for each vertex of~$D$
  and at least three such pairs for each extreme point,
  that is, at least $16+4\cdot 3=28$
  pairs of intersecting lines in total.  Therefore,
  $\binom{\rho}{2}\ge 28$. But if we have equality then any two
  lines of~$\mathcal{L}$ intersect.
  So all of them share a common plane or a common point. In the first
  case the drawing is contained in a single plane; in the
  second case the family $\mathcal L$ cannot cover the drawing. Thus
  $\binom{\rho}{2}> 28$, and, hence, $\rho\ge 9$.

  \ref{prop:lines_I} If the drawing of the icosahedron~$I$ is not
  contained in a
  single plane, then we can pick four extreme points of the convex
  hull of the drawing.  Each of these points represents a vertex of
  degree~5, so we need five lines to cover edges incident to this
  vertex, that is, $20$ lines in total, but we have double-counted the
  lines that go through pairs of the extreme points that we picked.
  Of these, there are at most ${4 \choose 2}=6$.  Thus we need at
  least $20-6=14$ lines to cover the drawing.

  \begin{figure}[htb]
    \centering
    \includegraphics{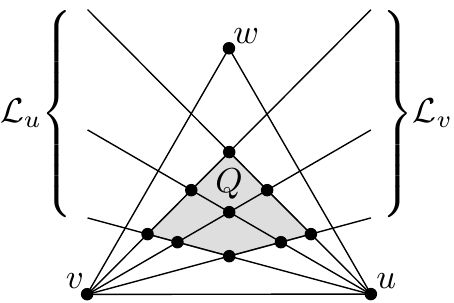}
    \caption{The families of lines $\mathcal{L}_u$ and $\mathcal{L}_v$}
    \label{fig:prop}
  \end{figure}

  Now assume that there exists a straight-line drawing of the
  icosahedron in a single plane covered by a family $\mathcal L$
  of twelve lines.  Let $u$, $v$, $w$ be the vertices of the outer
  face of~$I$.  Clearly, three distinct lines in~$\mathcal L$ form the
  triangle $uvw$.  For $s \in \{u,v,w\}$, we denote by $\mathcal{L}_s$
  the lines in~$\mathcal{L}$ that go through~$s$ and do not cover
  edges of the outer face.  Since $I$ is 5-regular, $|\mathcal{L}_s| =
  \deg(s)-2 = 3$.  Consider the set $P$ of intersection points between
  the line families $\mathcal{L}_u$ and $\mathcal{L}_v$.  The set $P$
  lies in the triangle $uvw$ and is bounded by the quadrilateral $Q$
  formed by the outer pairs of lines in $\mathcal{L}_v$ and
  $\mathcal{L}_u$; see Fig.~\ref{fig:prop}.

  The quadrilateral~$Q$ is convex and eight of the nine
  points in~$P$
  lie on the boundary of~$Q$, hence, for any line~$\ell$ in~$\mathcal
  L_w$, we have $|\ell \cap P| \le 3$.  Observe that $|\ell \cap P|=3$
  implies that $\ell$ goes through the only point of~$P$ that lies in
  the interior of~$Q$.  Thus the lines in~$\mathcal L_w$ can create at
  most seven triple intersection points with the lines
  in~$\mathcal{L}_u$ and~$\mathcal{L}_v$.

  The icosahedron is 5-regular, so all vertices must be placed at the
  intersection of at least three lines.  We need at least nine triple
  intersection points in order to place all $12-3$ inner vertices of
  the icosahedron---a contradiction.  \hfill\qed

\begin{proposition}
\label{prop:arcs}
\begin{enumerate*}[label=(\alph*)]
\item $\myarc(T) \,{\le}\, 3$;\label{prop:arcs_T}
\item $\myarc(O) \,{\le}\, 3$;\label{prop:arcs_O}
\item $\myarc(C) \,{\le}\, 4$;\label{prop:arcs_C}
\item $\myarc(D) \,{\le}\, 10$;\label{prop:arcs_D}
\item $\myarc(I) \le 7$. \label{prop:ico_arcs}
\end{enumerate*}
\end{proposition}
\begin{proof}
\nv{For the upper bounds for~\ref{prop:arcs_T}--\ref{prop:arcs_D}
see the drawings of the graphs in Figures~\ref{fig:tet_arc},
\ref{fig:oct_arc}, \ref{fig:cube_arc}, and~\ref{fig:dod_circle_arc}
respectively. While it is easy to see that these drawings
are valid, we argue more carefully that the icosahedron does indeed
admit a drawing with seven arcs.}
To construct the drawing in Fig.~\ref{fig:ico_circle_cover} (for
details see Fig.~\ref{fig:icosahedron-arc-cover}), we
first cover the edges of the icosahedron by seven objects,
grouped into a single cycle $K$ and two sets $L=\{L_0, L_1,
L_2\}$ and~$M=\{M_0, M_1, M_2\}$,
where $K$ is a cycle of length~6 and all elements of~$L$
and~$M$ are simple paths of length~4; see Fig.~\ref{fig:pathpart}.
We identify the paths and cycles with
their drawings as arcs and circles.
For a set $S\in\{\{K\}, L, M\}$ and a number $i\in\{0,1,2\}$,
let $(d_S,\alpha_{S_i})$ be the polar coordinates of the
center~$c(S_i)$ of the circle of radius~$r_S$ that covers
arc~$S_i \in S$ (see Fig.~\ref{fig:proof}).  We set the coordinates
and radii as follows:
\begin{align}
\alpha_K &= 0
& d_K &= 0
& r_K &= 1 \\
\alpha_{L_i} &= i \cdot {2\pi}/{3}
& d_L &= (3 + \sqrt{3})/{2}
& r_L &= \sqrt{5/2 + \sqrt{3}} \\
\alpha_{M_i} &= \pi/2+ i\cdot {2\pi}/{3}
& d_M &= (3 - \sqrt{3})/{2}
& r_M &= \sqrt{5/2 - \sqrt{3}}
\end{align}
Using 
the law of cosines, it is easy to compute the intersection points:
\begin{align}
\{A_i\} &:= L_i\cap L_{i+1}\cap  M_i  & \Rightarrow\;
A_i &= \big(i \cdot {2\pi}/{3}, ({1+\sqrt{3}})/{2}\big); \\
\{B_i\} &:= L_i\cap L_{i+1}\cap  K    & \Rightarrow\;
B_i &= (i \cdot {2\pi}/{3}, 1); \\
\{C_i\} &:= M_i\cap M_{i+2}\cap  K    & \Rightarrow\;
C_i &= (\pi/{3} + i \cdot {2\pi}/{3}, 1); \\
\{D_i\} &:= L_i\cap M_{i}\cap M_{i+1} & \Rightarrow\;
D_i &= \big({\pi}/{2} + i \cdot {2\pi}/{3}, (\sqrt{3}-1)/{2}\big).
\end{align}
For $i = 0, 1, 2$, let~$L_i$ be the larger arc of the covering circle
between the points $A_i$ and $B_i$, let~$M_i$ be the larger arc of the
covering circle between the points $C_{i+1}$ and $D_{i+2}$ (with
indices modulo~3), and let $K$ be the whole unit circle.
\end{proof}

\begin{figure}[tb]
  \begin{subfigure}[b]{0.385\textwidth}
      \centering
      \includegraphics{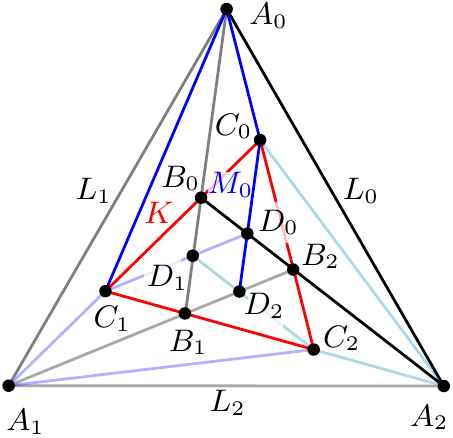}
      \caption{symmetric path partition: the black and gray arcs are
        in $L$, the light and dark blue arcs in $M$, and the red
        cycle is~$K$}
      \label{fig:pathpart}
    \end{subfigure}
    \hfill
    \begin{subfigure}[b]{0.6\textwidth}
      \centering
      \includegraphics{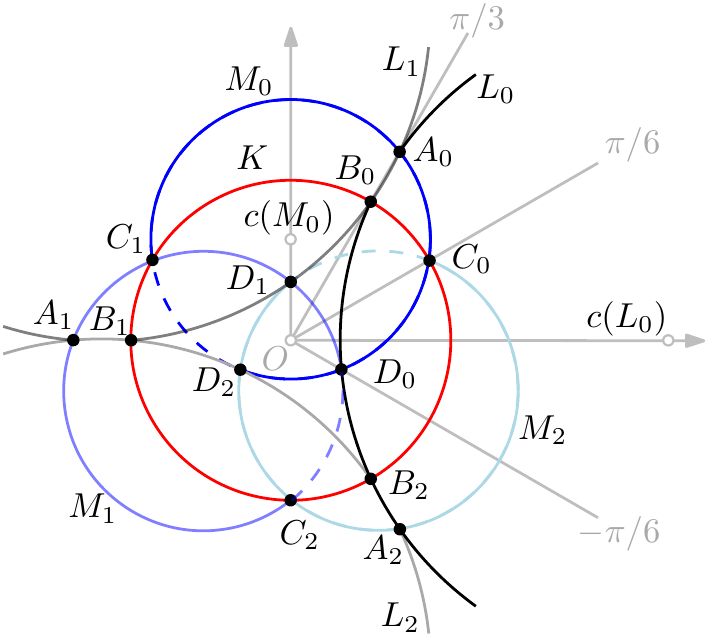}
      \caption{illustration of the proof of
        Proposition~\ref{prop:arcs}\ref{prop:ico_arcs}}
      \label{fig:proof}
    \end{subfigure}
    \caption{Bounding the arc number of the icosahedron}
    \label{fig:icosahedron-arc-cover}
\end{figure}

\section{Lower Bounds for $\sigma^1_d$}
\label{sec:lower}

Given a graph $G$, we obtain lower bounds for $\sigma^1_d(G)$ via
standard combinatorial characteristics of $G$ \nv{ in the same way as
for}
the bounds for
$\rho^1_d(G)$~\cite{cflrvw-dgflf-GD16}.
In particular, we prove a general lower bound
for $\sigma^1_d(G)$ in terms of the treewidth $\tw(G)$ of $G$,
which follows from the
fact that graphs with low parameter $\sigma^1_d(G)$ have small
separators.
This fact is interesting by itself and has yet another consequence:
graphs with bounded vertex degree can have a linearly large value
of~$\sigma^1_d(G)$ (hence, the factor of $n$ in the trivial bound
$\sigma^1_d(G)\le m\le n\cdot\Delta(G)/2$
is best possible).

We need the following definitions.
The \emph{linear arboricity} $\la(G)$ of a graph $G$ is the minimum
number of linear forests that partition the edge set of~$G$
\cite{h-cagI-ANYAS70}.  Let $W\subseteq V(G)$.
A set of vertices $S\subset V(G)$ is a \emph{balanced $W\!$-separator}
of
the graph~$G$ if $|W\cap C|\le|W|/2$ for every
connected component $C$ of $G - S$. Moreover, $S$ is a
\emph{strongly balanced
$W\!$-separator} if there is a partition $W\setminus S=W_1\cup W_2$
such that
$|W_i|\le|W|/2$ for both $i=1,2$ and there is no path between $W_1$
and $W_2$ that avoids~$S$.
Let $\sep_W(G)$ ($\sep^*_W(G)$) denote the minimum $k$ such that $G$
has
a (strongly) balanced $W\!$-separator $S$ with $|S|=k$.
Furthermore, let $\sep(G)=\sep_{V(G)}(G)$ and
$\sep^*(G)=\sep^*_{V(G)}(G)$.
Note that $\sep_W(G)\le\sep^*_W(G)$ for any $W \subseteq V(G)$ and,
in particular, $\sep(G)\le\sep^*(G)$.

It is known~\cite[Theorem 11.17]{FlumG06} that $\sep_W(G)\le\tw(G)+1$
for any $W\subseteq V(G)$.  On the other hand, $\tw(G)\le3k$ if
$\sep_W(G)\le k$ for every $W$ with $|W|=2k+1$.

Recall that
the bisection width $\bw(G)$ of a graph $G=(V,E)$ is the minimum
number of edges between two sets of vertices $W_1$ and $W_2$ with
$|W_1| =
\lceil n/2 \rceil$ and $|W_2| = \lfloor n/2 \rfloor$
partitioning~$V$.
Note that $\sep^*(G)\le\bw(G)+1$.

Now we use these graph parameters to \nv{bound the spherical cover
number from below}.  The proofs are similar to those regarding the
affine cover
number~\cite{cflrvw-dgflf-GD16}.
\nv{We restate Proposition~\ref{thm:rho_bw} as item (a) to make the following
theorem more self-contained.}

\begin{theorem}
\label{LowerSigma13}
For any integer $d \ge 1$ and any graph~$G$ with $n$ vertices
and $m$ edges, the
following bounds hold:
\begin{enumerate}[label=(\alph*)]
\item $\sigma^1_d(G)\ge \bw(G)/2$.
\label{thm:rho_bw_restated}

\item $\sigma^1_d(G)>n/10$ for almost all cubic graphs with $n$
  vertices; \label{thm:rho_cubic}

\item $\lceil \frac 32\sigma^1_d(G)\rceil \ge\la(G)$;
\label{thm:rho_la}

\item $\sigma^1_d(G)\ge\sep^*_W(G)/2$ for every $W\subseteq V(G)$;
\label{thm:rho_sep}

\item $\sigma^1_d(G) \ge \tw(G)/6$.
\label{thm:rho_tw}

\end{enumerate}
\end{theorem}

\begin{proof}
  \nv{
  For the  proof of \ref{thm:rho_bw_restated} see
  Proposition~\ref{thm:rho_bw}.
  }

  \ref{thm:rho_cubic}~The claim follows
  from~\ref{thm:rho_bw_restated} and from
  the fact that a random cubic graph on $n$ vertices has bisection
  width at least $n/4.95$ with probability $1-o(1)$~\cite{KostochkaM93}.

  \ref{thm:rho_la}~Given the drawing
  of the graph $G$ on $r=\sigma^1_d(G)$ circles, we remove an edge
  from each of the circles (provided such an edge exists), obtaining
  at (most) $r$ linear forests. The removed edges we group into
  (possible, degenerated) pairs, obtaining at most $\lceil r/2\rceil$
  additional linear forests. So, $\la(G)\le r+\lceil r/2\rceil$.

  \nv{\ref{thm:rho_sep} The proof is similar to Theorem~9(c) in~\cite{cflrvw-dgflf-GD16}.
 The difference of a factor of $1/2$ is due to the fact that a straight
 line pierces the plane at most once whereas and a circle pierces the hyperplane at most twice.
  }

  \nv{
  \ref{thm:rho_tw} follows from~\ref{thm:rho_sep} by the above mentioned
  relationship  between treewidth and balanced separators.
  }
\end{proof}

\begin{corollary}
  $\sigma^1_d(G)$ cannot be bounded from above by a function of
  $\la(G)$ or $v_{\ge 3}(G)$ or $\tw(G)$, where $v_{\ge 3}(G)$ is the
  number of vertices with degree at least 3.
\end{corollary}
\begin{proof}
  $\la(G)$: Akiyama et al.~\cite{AEH} showed that, for any cubic
  graph~$G$, $\la(G)=2$.  On the other hand, $v_{\ge 3}(G)=n$, so
  $\sigma^1_3(G)>\sqrt{n}$ by
  Proposition~\ref{prop:combinatorial_argument}\ref{thm:rho_mod_1}.
  Theorem~\ref{LowerSigma13}\ref{thm:rho_cubic} yields an even
larger gap.

  $v_{\ge 3}(G)$: Let $G$ be the disjoint union of $k$ cycles.  Then
  $v_{\ge 3}(G)=0$.  Clearly, an arrangement~$A$ of $\ell$ circles has
  at most $\ell^2$ vertices.  Each cycle of~$G$ ``consumes'' at least
  two vertices of~$A$ or a whole circle, so
  $\sigma^1_d(G)=\Omega(\sqrt{k})$.

  $\tw(G)$: Let $G$ be a caterpillar with linearly many vertices
  of degree~3.  Then, $\tw(G)=1$.  On the other hand, by
  Proposition~\ref{prop:combinatorial_argument}\ref{thm:rho_mod_1},
we
have
$\sigma^1_d(G)=\Omega(\sqrt
  n)$.
\end{proof}

\begin{lemma}
  \label{lem:nested_cycle2}
  A circular-arc drawing $\Gamma\subset\mathbb{R}$ of a graph~$G$ that
  contains~$k$ nested cycles cannot be covered by fewer than $k$
  circles.
\end{lemma}

\begin{proof}
  Fix any point inside the closed Jordan curve in~$\Gamma$ that
  corresponds to the innermost cycle of~$G$.  Let $\ell$ be an
  arbitrary line through this point. \nv{ Then $\ell$ crosses at
  least twice each of the Jordan curves that correspond to the
  nested cycles in~$G$.  Hence, there are at least $2k$ points where
  $\ell$ crosses~$\Gamma$.}

  On the other hand, consider any set of $r$ circles whose union
  covers~$\Gamma$.  Then it is clear that $\ell$ crosses each of these
  $r$ circles in at most two points, so there are at most $2r$ points
  where~$\ell$ crosses~$\Gamma$.  Putting together the two
  inequalities, we get $r\ge k$ as desired.
\end{proof}

At last we remark that there are graphs whose $\sigma_3^1$-value is a
lot smaller than their $\sigma_2^1$-value.

\begin{theorem}
  \label{thm:nested-triangles}
  For infinitely many $n$ there is a planar graph $G$ on $n$ vertices
  with $\sigma^1_2(G)=\Omega(n)$ and $\sigma^1_3(G)=O(n^{2/3})$.
\end{theorem}

\begin{proof}
  We use the same family $(G_i)_{i \ge 1}$ of graphs as Chaplick et
  al.~\cite[Theorem~24(b)]{cflrvw-dgflf-GD16} with $G_i=C_3 \times
  P_i$ and~$P_i$ a path with $i$ vertices.  Then~$G_i$ has $n_i=3i$
  vertices, $\rho^1_2(G_i)=\Omega(n_i)$, and
  $\rho^1_3(G_i)=O(n_i^{2/3})$.  The lower bound on $\sigma^1_2(G_i)$
  follows from Lemma~\ref{lem:nested_cycle2}.  The upper bound on
  $\sigma^1_3(G_i)$ follows from Proposition~\ref{prop:sigma-le-rho}
  for $l=1$ and $d=3$, which states that, for any graph~$G$,
  $\sigma^1_3(G) \le \rho^1_3(G)$.
\end{proof}

\section{An MIP Formulation for Estimating the Segment Number}
\label{sec:ilp}

In this section, we exploit an integer programming formulation for
\emph{locally consistent angle assignments}~\cite{dett-gd-99},
which we define below, to obtain lower bounds on the segment numbers
of planar graphs.  Our MIP
determines a locally consistent angle assignment with the maximum
number of $\pi$\--angles between incident edges.  Note that such angle
assignments are not necessarily realizable with straight-line edges in
the plane.  This is why the MIP yields only an upper bound for the
number of $\pi$-angles---and a lower bound for the segment number.
For the platonic graphs, however, it turns out that the bounds are
tight; see Tables~\ref{tab:seg_arc} and~\ref{tab:piangles}.

Let $G = (V, E)$ be a 3-connected graph with fixed embedding given by
a set~$\mathcal{F}$ of faces and an outer face~$f_0$.  For any vertex
$v\in
V$ and any face $f\in \mathcal{F}$, we introduce a \nv{fractional}
variable~$x_{v,f}
\in (0,2)$ whose value is intended to express the size of the angle
at~$v$ in~$f$, divided by~$\pi$.  Thus, $(\pi\cdot x_{v, f})_{v \in V, f
  \in \mathcal{F}}$ is an angle assignment for~$G$.  The following
constraints
guarantee that the assignment is locally consistent.  (For a
vertex~$v$
and a face~$f$, we write $v \sim f$ to express that $v$ is incident
to~$f$.)
\begin{align}
    \sum_{f\sim v} x_{v, f} &= 2 &&\text{for each } v \in V;\\
  \sum_{v\sim f} x_{v, f} &= \deg(f) - 2 \qquad
  &&\text{for each } f \in \mathcal{F} \setminus \{f_0\};\\
  \sum_{v\sim f_0} x_{v, f_0} &= \deg(f_0) + 2.\\
\intertext{For any vertex~$v$, let $L_v = \langle v_1, \dots, v_k
\rangle$ be the list of
vertices adjacent to $v$, in clockwise order as they appear in the
embedding.  
Due to the 3-connectivity of~$G$, any two vertices $v_t$ and
$v_{t+1}$ that are consecutive in~$L_v$ (and adjacent to~$v$) uniquely
define a face $f(v,t)$ incident to~$v$, $v_t$, and $v_{t+1}$.
For two vertices~$v_i$ and~$v_j$ with $i<j$, we express the angle
$\angle(v_ivv_j)$ as the sum of the angles at~$v$ in the faces
between~$v_i$ and~$v_j$.  As shorthand, we use
$y_{v,i,j}=\angle(v_ivv_j)/\pi \in (0,2)$:}
  y_{v,i,j} &=  \sum_{t=i}^{j-1} x_{v, f(v,t)} \qquad
  &&\text{for each } v \in V, \, 1 \le i < j \le \deg(v). \\
\intertext{We want to maximize the number of $\pi$\--angles between
any two edges
incident to the same vertex.  To this end, we introduce a 0--1
variable $s_{v,i,j}$ for any vertex~$v$ and $1 \le i<j \le \deg(v)$.
The intended meaning of $s_{v,i,j}=1$ is that $\angle(v_ivv_j)=\pi$.
We add the following constraints to the MIP:}
    s_{v,i,j} &\in \; \{0,1\} \\[-3.5ex]
    s_{v,i,j} &\le \; y_{v,i,j}
    \qquad\left.\rule{0ex}{5.5ex}\right\}
    &&\text{for each } v \in V, \, 1 \le i < j \le \deg(v).\\[-3.5ex]
    s_{v,i,j} &\le \; 2 - y_{v,i,j}
\intertext{If $y_{v,i,j} < 1$, the second constraint
will force $s_{v,i,j}$ to be~0 and the third constraint
will not be effective.  If $y_{v,i,j} > 1$, the third constraint
will force $s_{v,i,j}$ to be~0, and the second constraint
will not be effective.  Only if $y_{v,i,j} = 1$ (and
$\angle(v_ivv_j)=\pi$), both constraints will allow $s_{v,i,j}$ to
be~1.
This works because we want to maximize the total
number of $\pi$-angles between incident edges in a locally
consistent angle assignment. \nv{To this end, we use the following
objective:}}
\maximize & \sum_{v\in V} \; \sum_{1 \le i < j \le \deg(v)}
s_{v,i,j}.
\end{align}
Every $\pi$-angle between incident edges saves a segment; hence, in
any straight-line drawing of~$G$, the number of segments equals the
number of edges minus the number of $\pi$-angles.  In particular, this
holds for a drawing that minimizes the number of segments (and
simultaneously maximizes the number of $\pi$-angles).  Thus,
\begin{equation}
  \myseg(G) = |E| - \mypiang(G).
\end{equation}
\nv{Since 
\[\mypiang(G) \le \sum_{v\in V} \; \sum_{1 \le i < j \le
\deg(v)} s_{v,i,j},\] 
the above relationship provides the lower bound 
\[\myseg(G) \ge |E| - \sum_{v\in V} \; \sum_{1 \le i < j \le
\deg(v)} s_{v,i,j}\] for the segment number, which can be
computed by solving the MIP.}

Our MIP has $O(n^3)$ variables and constraints.  The experimental
results for the platonic graphs (which are 3-connected and thus have a
unique planar embedding) are displayed in Table~\ref{tab:piangles}.

\begin{table}
\newcolumntype{d}{D{.}{.}{3.1}}
\centering
\setlength{\tabcolsep}{2ex}
\nv{
\begin{tabular}{ldddd}
\toprule
graph $G$         & \multicolumn{1}{c}{octahedron}
& \multicolumn{1}{c}{cube}  & \multicolumn{1}{c}{dodecahedron}
& \multicolumn{1}{c}{icosahedron} \\
\midrule
$\mypiang(G)\le$       & 3      & 5               & 17     &  15  \\
$\myseg(G)$\hfill$\ge$ & 9      & 7               & 13     &  15  \\
\midrule
variables         & 60     & 48              & 120    &  180 \\
constraints       & 137    & 114             & 277    &  395 \\
runtime [s]       & 0.011    & 0.009             & 0.015    &
0.066 \\
\bottomrule
\end{tabular}
}
\medskip
\caption{Upper bounds on the number of $\pi$-angles and corresponding
  lower bounds on the segment numbers of the platonic graphs (except
  for the tetrahedron) obtained by the MIP and sizes of the MIP
  formulation for these instances.
  Running times were measured on a 64-bit machine with
  7.7 GB main memory and four Intel i5 cores with 1.90 GHz, using
  the MIP solver IBM ILOG CPLEX Optimization Studio 12.6.2.}
\label{tab:piangles}
\end{table}

\nv{In addition, to check the capabilities of the MIP, we have tested
  it on a family of triangulations~$(G_k)_{k\ge2}$ constructed by
  Dujmovi\'c et al.~\cite[Lemma~17]{Dujmovic2007DrawingsOP}; a variant
  of the nested-triangles graph (see
  Fig.~\ref{fig:thetriangulationsexact}).  The graph~$G_2$ is the
  octahedron.  For $k>2$, the triangulation $G_k$ is created by
  recursively nesting a triangle into the innermost triangle of
  $G_{k-1}$ and connecting its vertices to the vertices of the
  triangle it was nested into.  Note that $G_k$ has $n_k = 3k$
  vertices.  Dujmovi\'c et al.\ showed a lower bound of $2n_k-6$ on
  the segment number
  of~$G_k$ and a tight lower bound of $2n_k-3$ (see the proof
  of \cite[Lemma~17]{Dujmovic2007DrawingsOP} and
  Fig.~\ref{fig:thetriangulationsexact}) on the
  number of segments given the fixed embedding.  Figure~\ref{fig:chart}
  shows the runtime of the MIP in logarithmic scale for the
  triangulations $G_2,G_3,\dots,G_8$.  As expected, the runtime is (at
  least) exponential.  Interestingly, for each of these
  (embedded) graphs, our MIP finds a solution with $2n_k-3$
  ``segments'', thus matching the tight lower bound of Dujmovi\'c et
  al.\ for the fixed-embedding case.}

  \begin{figure}[tb]
    \begin{subfigure}[b]{0.39\textwidth}
      \centering
      \includegraphics[page=2]{g_n}
    \end{subfigure}
    \hfill
    \begin{subfigure}[b]{0.58\textwidth}
      \centering
      \includegraphics[width=\linewidth]{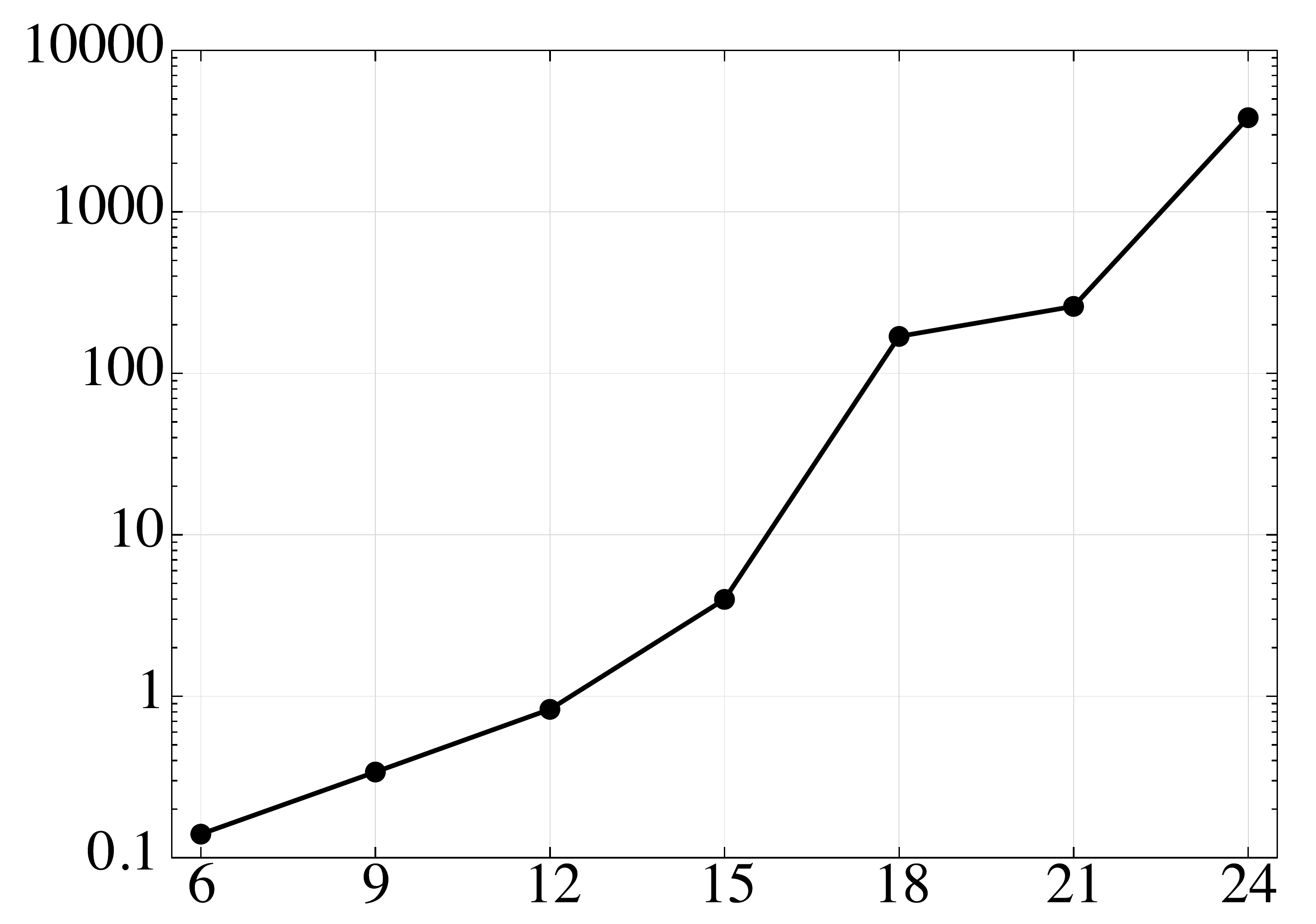}
    \end{subfigure}

    \begin{subfigure}[t]{0.38\textwidth}
      \centering
      \caption{optimal drawing of the triangulation~$G_k$ (with
        $n_k=3k$ vertices) of Dujmovi\'c et
        al.~\cite{Dujmovic2007DrawingsOP} using $2n_k-3$ segments}
       \label{fig:thetriangulationsexact}
    \end{subfigure}
    \hfill
    \begin{subfigure}[t]{0.58\textwidth}
      \centering
      \caption{runtime of the MIP applied to the graphs
        $G_2,G_3,\dots,G_8$ of Dujmovi\'c et
        al.~\cite{Dujmovic2007DrawingsOP}.  The numbers of vertices of
        the graphs are on the x-axis; the runtime in seconds is on the
        y-axis.  Note the log-scale at the y-axis.}
      \label{fig:chart}
    \end{subfigure}
    \caption{Testing the MIP: instances and runtime}
\end{figure}

\section{Discussion and Open Problems}
\label{sec:discussion}

\nv{As mentioned in the introduction, we now show that minimum-line
  drawings are indeed different from $\rho^1_2$-optimal drawings.
  Then we state some open problems regarding affine and spherical
  cover numbers.}

\begin{example}
\label{expl:rho_vs_ml}
\nv{
Minimum-line drawings are different from $\rho^1_2$-optimal drawings.
}
\end{example}
\begin{proof}
  \nv{We provide a graph $G$ with $\rho^1_2(G)=5$ and $\myseg(G)\le6$;
    see Fig.~\ref{fig:affine_cover}.  Then we show that every
    embedding of~$G$ on any arrangement of five straight lines
    consists of at least seven segments.}

  \begin{figure}[htb]
      \begin{subfigure}[b]{0.3\textwidth}
        \centering
        \includegraphics[page=1]{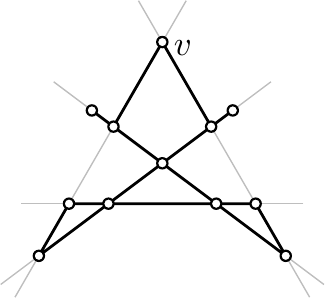}
        \caption{a $\rho^1_2$-optimal drawing of $G$ on 5 lines and 7
          segments}
        \label{fig:affine_cover}
      \end{subfigure}
      \hfill
      \begin{subfigure}[b]{0.33\textwidth}
        \centering
        \includegraphics[page=2]{ml_vs_rho}
        \caption{a minimum-line drawing of $G$ on 6 lines and 6 segments}
        \label{fig:min_line}
      \end{subfigure}
      \hfill
      \begin{subfigure}[b]{0.28\textwidth}
        \centering
        \includegraphics[page=3]{ml_vs_rho}
        \caption{a star-shaped arrangement of 5 straight lines}
        \label{fig:star}
      \end{subfigure}
      \caption{A graph~$G$ that shows that $\rho^1_2$-optimal drawing
        and minimum-line drawings are indeed different.}
  \end{figure}

  \nv{Chaplick et al.~\cite{cflrvw-dgflf-GD16} defined a vertex of a
    planar graph to be \emph{essential} if it has degree at least~3 or
    belongs to a cycle of length~3.  They observe that in any drawing
    of a graph any essential vertex is shared by two edges not lying
    on the same line.  Observe that~$G$ has nine essential vertices.
    Hence, any arrangement of straight lines that cover a drawing
    of~$G$ consists of at least five straight
    lines (with potentially ten intersection points).  Moreover, for
    the same reason, an arrangement of five straight lines covering
    a drawing of~$G$ must be \emph{simple}, that is, every two straight
    lines intersect and no three straight lines have a point in
    common.  There is only one such arrangement of five straight lines
    in the projective plane~\cite{gruenbaum}.  This combinatorially
    unique arrangement is star-shaped; see Fig.~\ref{fig:star}.}  

  \nv{The
    graph~$G$ has three triangles that are attached via one vertex in
    a chain-like fashion.  These triangles can only be embedded into
    faces of the arrangement; otherwise there would be a triangle that
    consumes two additional intersection points of the arrangement.
    Therefore, there is only one way to embed the three triangles on
    the arrangement, namely on some three consecutive spikes of the
    star.  This forces the degree-2 vertex~$v$ (see
    Fig.~\ref{fig:affine_cover}) to be on a bend (incident to two
    segments in the drawing) and makes the embedding combinatorially
    unique.  In this embedding of~$G$ we have seven segments, but
    $\myseg(G) \le 6$; see Fig.~\ref{fig:min_line}.}

  \nv{Finally, if a graph does not have a
    drawing with six segments covered by five straight lines in the
    projective plane, it also does not have one in the Euclidean plane,
    because we can embed a line arrangement in the Euclidean
    plane into one in the projective plane preserving the number of
    segments.  So we need at least six straight lines for a drawing
    with six segments.}
\end{proof}

\noindent
We close with some open problems.

\nv{We conjecture that our drawings in Figures~\ref{fig:dod_seg}
  and~\ref{fig:ico_seg_opt} are optimal.  This would mean that
  $\rho^1_3(D)=10$ and $\rho^1_3(I)=15$, but we have no proof for
  this.}

\nv{Is there a family~$\mathcal{G}$ of graphs such that the affine
  cover number $\rho^l_d(G)$ of every graph $G \in \mathcal{G}$ can be
  bounded by a function of the spherical cover number $\sigma^l_d(G)$?
For example in the plane (recall that
we only consider planar graphs there), $\rho^1_2(G) \in O(n)$, since
we can use a different line for each single edge, moreover, according
to Proposition~\ref{prop:combinatorial_argument}\ref{thm:rho_mod_1}
$\sigma^1_2(G) \in \Omega(\sqrt{n})$, therefore, we have that
$\rho^1_2(G) \in O(\sigma^1_2(G)^2)$.
For the given family of graphs
can this relation be tightened?
For example, Chaplick et al.~\cite[Example 22]{cflrvw-dgflf-GD16}
showed that there are triangulations for which
$O(\sqrt{n})$ lines suffice.  It would be even more interesting to
find families of graphs where there is an asymptotic difference
between the two cover numbers.}

We have already seen that $\sigma^2_3(K_n)$ grows asymptotically more
slowly than $\rho^2_3(K_n)$.
Is there a family of planar graphs where
$\sigma^1_2$ grows asymptotically more slowly than $\rho^1_2$?

\nv{Chaplick et al.~\cite{cflrvw-dgflf-GD16} showed that the hierarchy
  of affine cover numbers collapses in the following sense: For every
  graph~$G$, for every integer $d>3$, and for every integer~$l$ with
  $1 \le l \le d$, it holds that $\rho^l_d(G)=\rho^l_3(G)$.  The proof
  of this fact is based on affine maps, which transform planes into
  planes, but not spheres into spheres, so we don't know whether the
  hierarchy of spherical cover numbers collapses, too.}


\bibliographystyle{abbrvurl}
\bibliography{abbrv,refs}

\end{document}